\documentclass[aip,jmp,reprint]{revtex4-2}
\usepackage[breaklinks,colorlinks=true,allcolors={blue}]{hyperref}
\usepackage{amsmath,amsthm,amssymb,enumitem,mathtools,dsfont,graphicx,stmaryrd,ifthen,twoopt,tikz,float}

\usepackage[left=35mm,right=35mm,top=35mm,bottom=35mm]{geometry}
\tikzset{every picture/.style={line width=0.75pt}}

\newcommand{\restrict}[1]{\raise-.5ex\hbox{\ensuremath|}_{#1}}

\newcommand{\dd}[1]{\,\mathrm{d}#1}

\newcommand{\N}{\mathbb{N}}
\newcommand{\Z}{\mathbb{Z}}
\newcommand{\Q}{\mathbb{Q}}
\newcommand{\R}{\mathbb{R}}
\newcommand{\C}{\mathbb{C}}
\newcommand{\PPrb}[1]{\mathbb{P}\left[ #1\right]}

\newcommand{\bigo}{\mathcal{O}}
\renewcommand{\epsilon}{\varepsilon}

\newcommand{\norm}[1]{\left\| #1\right\|}
\newcommand{\abs}[1]{\left| #1\right|}
\newcommand{\floor}[1]{\left\lfloor #1 \right\rfloor}
\newcommand{\ceil}[1]{\left\lceil #1 \right\rceil}
\newcommand{\cc}[1]{\overline{#1}}
\newcommand{\hp}[2]{\left\langle #1 ,#2\right\rangle}

\newcommand{\id}{\mathrm{\,Id}}
\newcommand{\1}{\mathds{1}}
\newcommand{\vect}{\operatorname{Vect}}

\newcommand{\domain}[1]{\ifthenelse{\equal{#1}{}}{\mathcal{D}}{\mathcal{D}\left(#1\right)}}

\theoremstyle{plain}
\newtheorem{Theorem}{Theorem}
\newtheorem{Lemma}[Theorem]{Lemma}

\newtheorem{Corollary}[Theorem]{Corollary}
\newtheorem{Conjecture}{Conjecture}

\theoremstyle{definition}

\theoremstyle{remark}

\newtheorem*{Remarks}{Remarks}

\begin{document}
\title{Sharp Bounds for the Integrated Density of States of a Strongly Disordered 1D Anderson-Bernoulli Model}
\author{Daniel S\'anchez-Mendoza}
\email{dsanchezmendoza@unistra.fr}
\affiliation{Université de Strasbourg, Institut de Recherche Mathématique Avancée UMR 7501, F-67000 Strasbourg, France.}
\date{July, 2021}

\begin{abstract}
In this article we give upper and lower bounds for the integrated density of states (IDS) of the 1D discrete Anderson-Bernoulli model when the disorder is strong enough to separate the two spectral bands. These bounds are uniform on the disorder and hold over the whole spectrum. They show the existence of a sequence of energies in which the value of the IDS can be given explicitly and does not depend on the disorder parameter.
\end{abstract}

\maketitle

\section{Introduction and Results}
Bounds found in the literature for the integrated density of states (IDS) of the Anderson model on $\Z^d$ are usually given in order to prove Lifschitz tails. Naturally, the greater the generality on the dimension and the potential distribution the less precise and explicit the bounds. The most general proof of Lifschitz tails at the bottom of the spectrum was given by Simon \cite{Simon} using Dirichlet-Neumann bracketing. Proper to the 1D case, Schulz-Baldes \cite{Schulz} worked on periodic infinite Jacobi matrices plus an independent identically distributed (i.i.d.) random potential with a discrete distribution employing the transfer matrix approach. Not only did he prove Lifschitz tails at all spectral band edges, but he also gave the Lifschitz constant (defined as in Ref. \onlinecite[Eq. 4.45]{Aizenman}) as a function of the potential distribution and the IDS of the free Jacobi matrix. More recently, in Ref. \onlinecite{Bishop} the authors obtained the value of the Lifschitz constant at the bottom of the spectrum of the 1D Anderson-Bernoulli model by approximating eigenfunctions associated to low eigenvalues with sine waves supported on the 0's of the potential. Since the goal of these articles, and many others for that matter, was to prove the existence of Lifschitz tails, the bounds for the IDS they obtained only hold near the edges of the bands, and not in the bulk.

In this article we derive bounds for the IDS of the 1D Anderson-Bernoulli model that hold throughout the first band of the spectrum and do not depend on the disorder parameter, as long as there are two disjoint spectral bands. The upper and lower bounds are sharp in the sense that they coincide on a countably infinite set, giving rise to two sequences of ``special" energies at which the value of the IDS is completely explicit and independent of the disorder parameter. The bounds also give the Lifschitz behaviour, as well as the rate at which the IDS converges uniformly to the IDS of the free Laplacian when the Bernoulli parameter goes to $0$.

The random Schrodinger operator that concerns us is
\begin{align*}
H_{p,\zeta}\coloneqq-\Delta+\zeta V_p:\ell^2(\N)&\longrightarrow\ell^2(\N)\\
\phi&\longmapsto (H_{p,\zeta}\phi)(j)\coloneqq\left[2\phi(j)-\phi(j+1)-\phi(j-1)\right]+\zeta V_p(j)\phi(j),
\end{align*}
where the Laplacian has the Dirichlet boundary condition $\phi(0)=0$ and the potential $\{V_p(j)\}_{j\in\N}$ is an i.i.d. sequence of random variables defined over a probability space $(\Omega,\mathcal{F},\mathbb{P})$ following a non-degenerate Bernoulli($p$) distribution ($\PPrb{V_p(j)=1}=p=1-\PPrb{V_p(j)=0}$). We chose to define $H_{p,\zeta}$ on $\N$ and not $\Z$ since it makes no difference on the IDS.

We restrict ourselves to the case where the disorder parameter $\zeta$ is at least $4=\norm{-\Delta}$ (for the weak disorder regime we refer to Refs. \onlinecite{Campanino,Jitomirskaya}). Since the almost sure spectrum of $H_{p,\zeta}$ is
\begin{equation*}
    \sigma(H_{p,\zeta})=[0,4]+\{0,\zeta\}=[0,4]\cup[\zeta,\zeta+4],\qquad \mathbb{P}\text{-a.s.},
\end{equation*}
the condition $\zeta\geq4$ guaranties no overlap, except possibly for a single point, between $[0,4]$ and $[\zeta,\zeta+4]$. We will refer to the previous intervals as the first and second band of the spectrum, even for $\zeta=4$ when there is no actual spectral gap.

For the IDS  of $H_{p,\zeta}$, denoted $I_{p,\zeta}$, we use the eigenvalue counting definition (there are other equivalent ones)
\begin{equation*}
   I_{p,\zeta}(x)\coloneqq\lim_{L\to \infty}\frac{1}{L}\#\left\{\lambda\in\sigma\left(H_{p,\zeta}\restrict{\ell^2\left(\{1,\ldots,L\}\right)}\right)\,\middle| \,\lambda\leq x\right\},\qquad\mathbb{P}\text{-a.s.},
\end{equation*}
and recall that $I_{p,\zeta}$ is a non-random continuous \cite{Delyon} function.

Before stating our main result we set some notation. The floor and ceiling functions are denoted $\floor{\cdot}$ and $\ceil{\cdot}$ respectively, $\N=\{1,2,\ldots\}$, $\N_0=\N\cup\{0\}$ and
\begin{equation*}
    \beta(x)\coloneqq\frac{\pi}{2\arcsin\left(\sqrt{x}/2\right)},\qquad x\in(0,4].
\end{equation*}

\begin{Theorem}
Let $\zeta\geq 4$, then
\begin{align}
p\sum_{k=1}^\infty (1-p)^{\ceil{k\beta(x)}-1}&\leq\,I_{p,\zeta}(x)\leq p\sum_{k=1}^\infty (1-p)^{\floor{k\beta(x)}-1},&x\in(0,2],\label{B1}\\
p\sum_{k=1}^\infty(1-p)^{\ceil{k\beta(x)}}&\leq(1-p)-I_{p,\zeta}(4-x)\leq p\sum_{k=1}^\infty (1-p)^{\floor{k\beta(x)}},&x\in(0,2].\label{B2}
\end{align}
\end{Theorem}
\begin{Remarks}\leavevmode
\begin{enumerate}
\item The right and left sides of \eqref{B1} and \eqref{B2} do not depend on $\zeta$. The functions $\floor{\cdot}$ and $\ceil{\cdot}$ make them discontinuous at every $x$ such that $\beta(x)\in\Q$.
\item Both \eqref{B1} and \eqref{B2} hold on $(0,4)$, we chose to state them on $(0,2]$ because after $x=2$ they are not optimal bounds. This will be better explained at the end of the proof (see \eqref{B3}).
\end{enumerate}
\end{Remarks}

The proof of Theorem 1 consist of constructing new operators that bound $H_{p,\zeta}$ from above and below by subdividing the underling lattice. This subdividing technique, known as Dirichlet-Neumann  bracketing, was used by Simon \cite{Simon} to prove Lifschitz tails for the Anderson model in $\Z^d$. He subdivided the lattice in cubes of size $c\,x^{-1/2}$ for some small $x$ and then used probabilistic bounds on the smallest eigenvalue of the Anderson model operator restricted to such cubes to obtain bounds on the IDS at $x$. In contrast, we follow Ref. \onlinecite{Bishop} by subdividing according to the connected components of $0$'s of the potential. With this way of subdividing, the resulting operators have explicit eigenvalues, which in turn allows for an explicit computation of their IDS at every energy $x$. We will derive one lower and two upper bounds of $I_{p,\zeta}$. The lower bound is given by the Cauchy Eigenvalue Interlacing Theorem as we delete all points with positive potential.  The two upper bounds follow from two applications of the Neumann part of Dirichlet-Neumann bracketing. For the first one, we apply it directly to (the finite volume restriction of) $H_{p,\zeta}$ in order to decouple all the connected components of zeros of the potential. For the second one, we first show that doubling every positive potential point and halving the disorder parameter $\zeta$ results in a lower operator to which the Neumann part of Dirichlet-Neumann bracketing is then applied in order to decouple every connected component of zeros of the potential together with its two positive potential boundary points.

The first consequence of Theorem 1 comes from the fact that whenever $\beta(x)\in\N$ the right and left sides of \eqref{B1} coincide, and of course, do not depend on $\zeta$ (the same can be said for \eqref{B2}). Simplifying the resulting geometric series we obtain:
\begin{Corollary}
Let $\zeta\geq 4$, then for $n\in\N\setminus\{1\}$
\begin{equation*}
I_{p,\zeta}(\beta^{-1}(n))=\frac{p}{(1-p)}\frac{(1-p)^n}{1-(1-p)^n},\qquad I_{p,\zeta}(4-\beta^{-1}(n))=(1-p)-\frac{p(1-p)^n}{1-(1-p)^n}.
\end{equation*}
\end{Corollary}

To the best of the author's knowledge, Corollary 2 is the first instance in which the value of the IDS of any Anderson model has been given explicitly at non-trivial energies, with the exception of the Cauchy distribution \cite{Lloyd}. Notice that the sequence of special energies $\{\beta^{-1}(n)\}_{n\geq2}$ (resp. $\{4-\beta^{-1}(n)\}_{n\geq2}$) starts at $2$ and then decreases (resp. increases) strictly approaching its limit $0$ (resp. $4$). This, together with the continuity of $I_{p,\zeta}$, gives the rather obvious $I_{p,\zeta}(0)=0$ and $I_{p,\zeta}(4)=1-p$.

Corollary 2 implicitly introduces the piece-wise smooth function 
\begin{equation*}
    f_p(x)\coloneqq\begin{cases}\frac{p}{(1-p)}\frac{(1-p)^{\beta(x)}}{1-(1-p)^{\beta(x)}},&x\in(0,2],\\(1-p)-\frac{p(1-p)^{\beta(4-x)}}{1-(1-p)^{\beta(4-x)}},&x\in[2,4),
    \end{cases}
\end{equation*}
which satisfies $f_p(x)=I_{p,\zeta}(x)$ whenever $x$ is a special energy. Computing numerically $I_{p,4}$ suggests that $I_{p,4}\geq f_p$ (resp. $I_{p,4}\leq f_p$) in a neighborhood to the left (resp. right) of every special energy. We cannot say if the same property holds for every $\zeta>4$ since $I_{p,\zeta}$ rapidly decreases as $\zeta$ increases. These numerics do indicate that $I_{p,\zeta}$ converges to the lower bound as $\zeta$ goes to infinity, however this convergence cannot be uniform since the proposed limit is discontinuous (see Fig. \ref{Fig1}).
\begin{figure}[H]
    \centering
    \includegraphics[width=12cm]{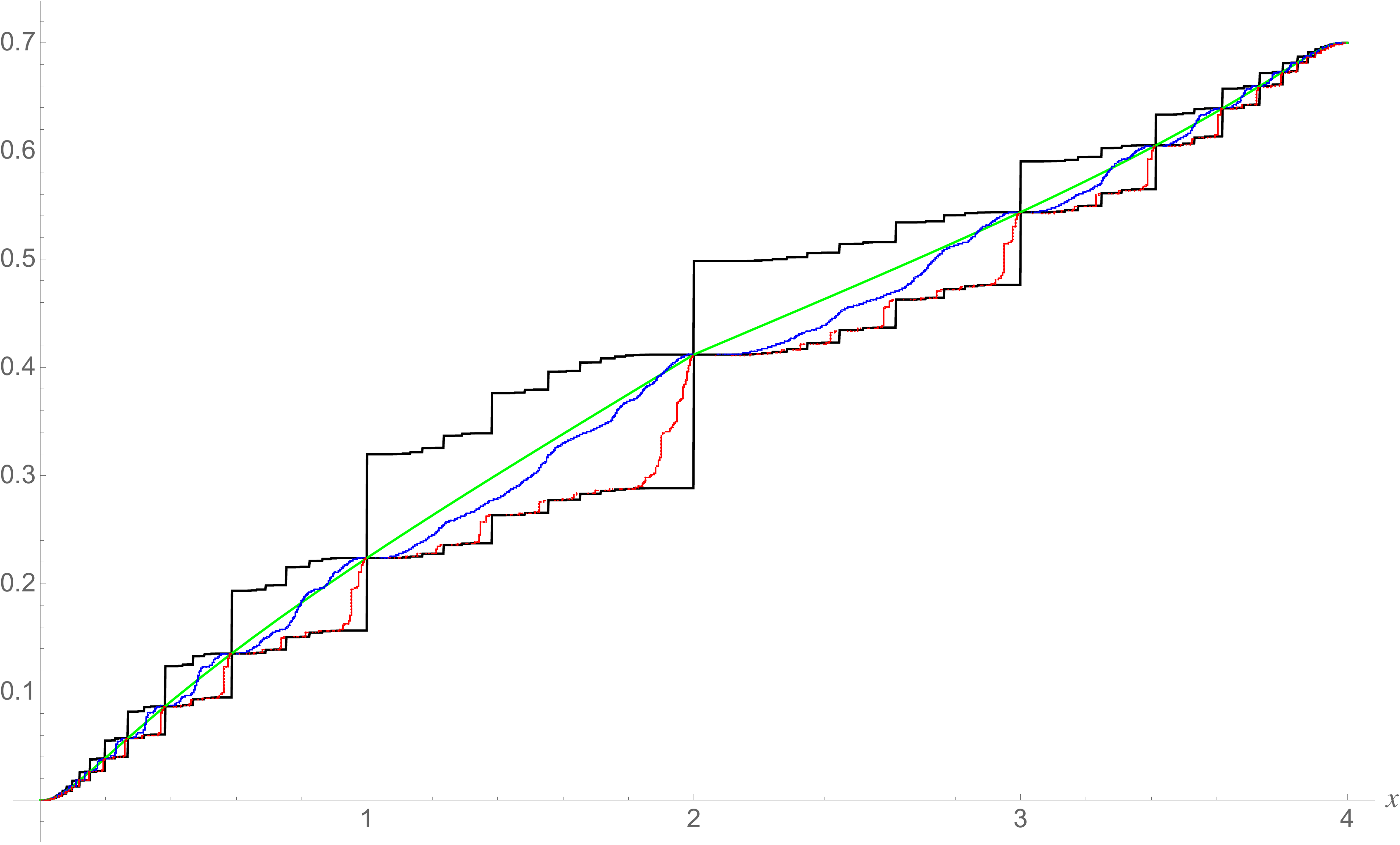}
    \caption{\textcolor{black}{\textbf{---}} Upper/Lower bound from Theorem 1\quad\textcolor{green}{\textbf{---}} $f_p$ \quad\textcolor{blue}{\textbf{---}} $I_{p,4}$ \quad\textcolor{red}{\textbf{---}} $I_{p,20}$\qquad(all for $p=3/10$).\\The special energies are the points $x$ at which all plots intersect. $I_{p,4}$ and $I_{p,20}$ were computed numerically from a $10^5\times10^5$ matrix.}
    \label{Fig1}
\end{figure}

To recover the Lifshitz behaviour from Theorem 1 we replace the discontinuous bounds with continuous ones by using the trivial estimate $y-1\leq\floor{y}\leq y\leq\ceil{y}\leq y+1$ and the geometric series:
\begin{Corollary}[Lifshitz Tails]
Let $\zeta\geq 4$, then
\begin{align}
&p\frac{(1-p)^{\beta(x)}}{1-(1-p)^{\beta(x)}}\leq I_{p,\zeta}(x)\leq \frac{p}{(1-p)^2}\frac{(1-p)^{\beta(x)}}{1-(1-p)^{\beta(x)}},& x\in(0,2]\label{IUC},\\
&p(1-p)\frac{(1-p)^{\beta(x)}}{1-(1-p)^{\beta(x)}}\leq (1-p)-I_{p,\zeta}(4-x)\leq \frac{p}{1-p}\frac{(1-p)^{\beta(x)}}{1-(1-p)^{\beta(x)}},& x\in(0,2].
\end{align}
Since $\beta(x)=\frac{\pi}{\sqrt{x}}+\bigo(\sqrt{x})$ for $x\to0$, this gives the Lifshitz tails at both edges of the first band as well as the Lifshitz constants
\begin{equation*}
\pi\ln(1-p)=\lim_{x\downarrow0}\sqrt{x}\ln I_{p,\zeta}(x),\qquad
\pi\ln(1-p)=\lim_{x\uparrow4}\sqrt{4-x}\ln\big[I_{p,\zeta}(4)-I_{p,\zeta}(x)\big].
\end{equation*}
\end{Corollary}

An asymptotically equivalent version of \eqref{IUC} was already given in Ref. \onlinecite{Bishop}, an article whose ideas have inspired this one, particularly Sec. II.

Our last Corollary simply gives the uniform distance between $I_{p,\zeta}$ and $I_{0,0}$ (the IDS of the 1D free Laplacian):
\begin{Corollary}
Let $\zeta\geq 4$, then $\displaystyle\sup_{x\in[0,4]}\abs{I_{p,\zeta}(x)-I_{0,0}(x)}=\abs{I_{p,\zeta}(4)-I_{0,0}(4)}=p$.
\end{Corollary}

Theorem 1 and its Corollaries refer to the first band of $\sigma(H_{p,\zeta})$ but they can be transferred to the second one by means of the standard unitary map $(U\phi)(j)=(-1)^j\phi(j)$. Indeed, $U$ transforms $H_{p,\zeta}$ as $U H_{p,\zeta} U^*=4+\zeta-\left(-\Delta+\zeta[1-V_p]\right)$. Since $\{1-V_p(j)\}_{j\in\N}$ are i.i.d. following a Bernoulli($1-p$) distribution we have for all $\zeta\geq0$
\begin{equation*}
    I_{p,\zeta}(x)=1-I_{1-p,\zeta}(4+\zeta-x),\qquad x\in\R.
\end{equation*}
This equality exchanges the two bands of the spectrum of $H_{p,\zeta}$ and therefore allows us to restate our results on the second one, which will now depend on $\zeta$ by translation.
 
The rest of the article is organized as follows: In Sec. II we construct from $H_{p,\zeta}$ four new random operators whose IDS we compute explicitly. In Sec. III we complete the proof Theorem 1 by comparing the eigenvalues of the new operators to those of $H_{p,\zeta}$ in the finite volume restriction, and we also prove Corollary 4. Finally, we give some closing remarks and a couple of conjectures in Sec. IV.

\section{IDS of Some Random Operators}
The gap of the spectrum of $H_{p,\zeta}$ suggests that eigenvalues of the first band come from eigenvectors whose mass is concentrated on the points where the potential is $0$. With this in mind we define the random variables
\begin{align*}
    L_1&\coloneqq\min\{j>0\,|\,V_p(j)=1\},& Y_1&\coloneqq L_1-1,\\
    L_{n+1}&\coloneqq \min\{j>L_n\,|\,V_p(j)=1\},& Y_{n+1}&\coloneqq L_{n+1}-L_n-1.
\end{align*}
In words, $L_n$ is the position of the $n$-th 1 and $Y_{n+1}$ is the number of 0's between the $n$-th and $(n+1)$-th 1. Fig. \ref{Fig2} shows a possible realization of $H_{p,\zeta}$ (we will keep this realization as the starting point for all later diagrams):
\begin{figure}[H]
    \centering
    \includegraphics{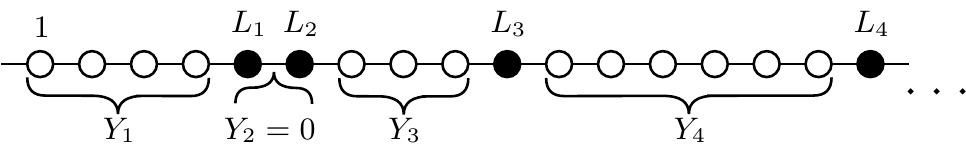}
    \caption{A possible realization of $H_{p,\zeta}$. The white (resp. black) dots represent points where $V_p(j)=0$ (resp. $V_p(j)=1$).}
    \label{Fig2}
\end{figure}
By definition, the $Y_i$ are i.i.d. following a geometric distribution $\PPrb{Y_i=k}=(1-p)^k p$ for $k\in \N_0$, and $L_n=n+\sum_{i=1}^n Y_i$. We also define $\cc{Y}_n\coloneqq\max_{1\leq i\leq n}Y_i$ and prove it grows logarithmically in the following lemma.
\begin{Lemma}
$\displaystyle \frac{\cc{Y}_n}{\ln(n)/\abs{\ln(1-p)}}\xrightarrow[n\to\infty]{\mathbb{P}\text{-a.s.}}1$.
\end{Lemma}
\begin{proof}
We start with the $\liminf$. For any $\epsilon>0$ we have 
\begin{align*}
    \PPrb{\cc{Y}_n\leq\frac{(1-\epsilon)\ln (n)}{\abs{\ln (1-p)}}}&=\PPrb{Y_1\leq\frac{(1-\epsilon)\ln (n)}{\abs{\ln (1-p)}}}^n=\PPrb{Y_1\leq\floor{\frac{(1-\epsilon)\ln (n)}{\abs{\ln (1-p)}}}}^n\\
    &=\left(1-(1-p)^{\floor{\frac{(1-\epsilon)\ln (n)}{\abs{\ln (1-p)}}}+1}\right)^n\\
    &\leq\left(1-(1-p)^{\frac{(1-\epsilon)\ln (n)}{\abs{\ln (1-p)}}+1}\right)^n=\left(1-\frac{(1-p)}{n^{1-\epsilon}}\right)^n,
\end{align*}
which is summable by comparison with $e^{-(1-p)n^\epsilon}$. With this we can apply the Borel-Cantelli Lemma and then intersect over all $\epsilon\in\Q$ to obtain $\liminf_{n\to\infty}\frac{\cc{Y}_n}{\ln(n)/\abs{\ln(1-p)}}\geq1$ almost surely.

Following the same reasoning,
\begin{align*}
    \PPrb{Y_n\geq\frac{(1+\epsilon)\ln (n)}{\abs{\ln (1-p)}}}&=\PPrb{Y_n\geq\ceil{\frac{(1+\epsilon)\ln (n)}{\abs{\ln (1-p)}}}}=(1-p)^{\ceil{\frac{(1+\epsilon)\ln (n)}{\abs{\ln (1-p)}}}}\\&\leq(1-p)^{\frac{(1+\epsilon)\ln (n)}{\abs{\ln (1-p)}}}=\frac{1}{n^{1+\epsilon}}
\end{align*}
implies $\limsup_{n\to\infty}\frac{Y_n}{\ln(n)/\abs{\ln(1-p)}}\leq1$ almost surely.

Now, for every $\omega$ in the full probability event $\{\limsup_{n\to\infty}\frac{Y_n}{\ln(n)/\abs{\ln(1-p)}}\leq1\}$ and every $\epsilon>0$ there exists $k(\omega,\epsilon)\in\N$ such that $j\geq k(\omega,\epsilon)$ implies $\frac{Y_j(\omega)}{\ln(j)/\abs{\ln(1-p)}}\leq1+\epsilon$, hence 
\begin{align*}
\frac{\cc{Y}_n(\omega)}{\ln(n)/\abs{\ln(1-p)}}&= \frac{\abs{\ln(1-p)}}{\ln(n)}\max_{1\leq j\leq n}Y_j(\omega)\\
&\leq\frac{\abs{\ln(1-p)}}{\ln(n)}\left(\max_{1\leq j< k(\omega,\epsilon)}Y_j(\omega)\right)+\frac{\abs{\ln(1-p)}}{\ln(n)}\max_{k(\omega,\epsilon)\leq j\leq n}Y_j(\omega)\\
&\leq\frac{\abs{\ln(1-p)}}{\ln(n)}\left(\max_{1\leq j< k(\omega,\epsilon)}Y_j(\omega)\right)+\max_{k(\omega,\epsilon)\leq j\leq n}\frac{Y_j(\omega)}{\ln(j)/\abs{\ln(1-p)}}\\&\leq\frac{\abs{\ln(1-p)}}{\ln(n)}\left(\max_{1\leq j< k(\omega,\epsilon)}Y_j(\omega)\right)+(1+\epsilon)\xrightarrow[n\to\infty]{}1+\epsilon,
\end{align*}
which finishes the proof.
\end{proof}

Using the $Y_i$'s we will define four random operators whose eigenvalues will be compared in Sec. III with those of $H_{p,\zeta}$ (in the finite volume restriction) either by direct comparison of operators and/or by the Min-Max Principle.

The first random operator is $\bigoplus_{i=1}^\infty-\Delta_{Y_i}$, where (whenever $n>0$) $-\Delta_{n}$ is the usual $n\times n$ Laplacian matrix (see Ref. \onlinecite[Section 2.2.3.3]{Thomas}):
\begin{equation*}
-\Delta_{n}\coloneqq\begin{pmatrix}
2 & -1 & & \\
-1 & \ddots & \ddots & \\
& \ddots & \ddots & -1 \\
 & & -1 & 2 \end{pmatrix},\qquad\sigma(-\Delta_{n})=\left\{\lambda_{n,k}\coloneqq 4 \sin^2\left(\frac{\pi k}{2(n+1)}\right)\,\middle|\,k=1,\ldots,n\right\}.
\end{equation*}
As we shall see, its IDS (denoted by $I_p^\leq$), is not continuous, so we also introduce its limit from the left (denoted by $I_p^<$) and call them both IDS. These two IDS are given by the deterministic $\mathbb{P}$-a.s. limits
\begin{align*}
    I_p^\leq(x)&\coloneqq\lim_{n\to \infty}\frac{1}{L_n}\#\left\{\lambda\in\sigma\left(\bigoplus_{i=1}^n-\Delta_{Y_i}\right)\,\middle| \,\lambda\leq x\right\},\\
    I_p^<(x)&\coloneqq\lim_{n\to \infty}\frac{1}{L_n}\#\left\{\lambda\in\sigma\left(\bigoplus_{i=1}^n-\Delta_{Y_i}\right)\,\middle| \,\lambda<x\right\}.
\end{align*}
Since $\frac{\sum_{i=1}^n Y_i}{L_n}\xrightarrow[n\to\infty]{\mathbb{P}\text{-a.s.}}1-p$, dividing by $L_n$ instead of the more natural $\sum_{i=1}^n Y_i$ amounts to a constant factor. Clearly $I_p^\leq(0)=I_p^<(0)=0$ and $I_p^\leq(4)=I_p^<(4)=1-p$. For $I_p^<$ and $x\in(0,4)$ we have
\begin{align*}
    \#\left\{\lambda\in\sigma\left(\bigoplus_{i=1}^n-\Delta_{Y_i}\right)\,\middle| \,\lambda<x\right\}&=\sum_{k=1}^{\cc{Y}_n} \#\left\{1\leq i\leq n\,\middle|\,Y_i>0,\,k\leq Y_i,\,\lambda_{Y_i,k}<x\right\}\\
    &=\sum_{k=1}^{\cc{Y}_n} \#\left\{1\leq i\leq n\,\middle|\,k\beta(x)-1<Y_i\right\}.
\end{align*}
The law of large numbers suggests that $\#\left\{1\leq i\leq n\,\middle|\,y<Y_i\right\} \approx n\PPrb{y<Y_1}$ for large $n$, so we set
\begin{equation*}
    R(y,n)=\#\left\{1\leq i\leq n\,\middle|\,y<Y_i\right\}- n\PPrb{y<Y_1}
\end{equation*}
and claim that $\frac{\sup_{y\in\R}\abs{R(y,n)}}{n^\gamma}\xrightarrow[n\to\infty]{\mathbb{P}\text{-a.s.}}0$ for $\gamma>1/2$. The proof of this claim only requires an application of the Dvoretzky–Kiefer–Wolfowitz inequality \cite{Massart}. Indeed
\begin{align*}
    \PPrb{\frac{\sup_{y\in\R}\abs{R(y,n)}}{n^\gamma}>\epsilon}&=\PPrb{\sup_{y\in\R}\abs{\frac{\#\{1\leq i\leq n\,|\,y<Y_1\}}{n}-\PPrb{y<Y_1}}>\epsilon n^{\gamma-1}}\\
    &\leq2\exp\left(-2\epsilon^2 n^{2\gamma-1}\right)
\end{align*}
which is summable, so the Borel-Cantelli Lemma gives the claim. Since $\cc{Y}_n$ grows like $\bigo(\ln(n))$ and $\frac{L_n}{n}\xrightarrow[n\to\infty]{\mathbb{P}\text{-a.s.}}\frac{1}{p}$ we have
\begin{equation*}
    \frac{1}{L_n}\abs{\sum_{k=1}^{\cc{Y}_n}R(k\beta(x)-1,n)}\\\leq\frac{\cc{Y}_n}{L_n}\sup_{y\in\R}\abs{R(y,n)}\xrightarrow[n\to\infty]{\mathbb{P}\text{-a.s.}}0,
\end{equation*}
hence we have shown 
\begin{align*}
I_p^<(x)&=p\sum_{k=1}^\infty \PPrb{k\beta(x)-1<Y_1}\\
&=p\sum_{k=1}^\infty \PPrb{\floor{k\beta(x)}-1<Y_1}=p\sum_{k=1}^\infty (1-p)^{\floor{k\beta(x)}},\qquad x\in(0,4).
\end{align*}
The computation for $I_p^\leq$ follows the same steps, so we just make explicit the differences:
\begin{align*}
    I_p^\leq(x)&=p\sum_{k=1}^\infty \PPrb{k\beta(x)-1\leq Y_1}\\
    &=p\sum_{k=1}^\infty \PPrb{\ceil{k\beta(x)}-1\leq Y_1}=p\sum_{k=1}^\infty (1-p)^{\ceil{k\beta(x)}-1},\qquad x\in(0,4).
\end{align*}
It is worth noting that if $x\in(0,4)$ is such that $\beta(x)\in(1,\infty)\setminus\Q$ then $I_p^\leq(x)= I_p^<(x)$. Moreover, since $U(-\Delta_{n})U^*=4+\Delta_{n}$ we have
\begin{equation*}
    \#\{1\leq k\leq n\,|\,\lambda_{n,k}\leq x\}= \#\{1\leq k\leq n\,|\,4-x\leq \lambda_{n,k}\}=n-\#\{1\leq k\leq n\,|\,\lambda_{n,k}< 4-x\},
\end{equation*}
which implies
\begin{equation}\label{ET1}
I_p^\leq(x)=(1-p)-I_p^<(4-x),\qquad x\in\R.
\end{equation}

We now move to the IDS of other random operators with a similar structure. For $n>0$ we define the $n\times n$ matrix $A_{n}(i,j)\coloneqq\delta_{1,i}\delta_{1,j}+\delta_{n,i}\delta_{n,j}$, and introduce the Dirichlet and Neumann Laplacian matrices with their spectra (see Ref. \onlinecite[Theorem 2.4]{Simon}):
\begin{align*}
    -\Delta_n^D&=-\Delta_n+A_n,\qquad\sigma(-\Delta_{n}^D)=\left\{\lambda^D_{n,k}\coloneqq 4 \sin^2\left(\frac{\pi k}{2n}\right)\,\middle|\,k=1,\ldots,n\right\},\\
    -\Delta_n^N&=-\Delta_n-A_n,\qquad\sigma(-\Delta_{n}^N)=\left\{\lambda^N_{n,k}\coloneqq 4 \sin^2\left(\frac{\pi(k-1)}{2n}\right)\,\middle|\,k=1,\ldots,n\right\}.
\end{align*}
These Dirichlet and Neumann Laplacians are related by $U(-\Delta_n^D)U^*=4+\Delta_n^N$, therefore
\begin{equation}\label{ET2}
I_p^{\leq,D}(x)=(1-p)-I_p^{<,N}(4-x),\qquad x\in\R,
\end{equation}
where we have used the (natural) definitions
\begin{align*}
I_p^{\leq,D}(x)&\coloneqq\lim_{n\to \infty}\frac{1}{L_n}\#\left\{\lambda\in\sigma\left(\bigoplus_{i=1}^n-\Delta^D_{Y_i}\right)\,\middle| \,\lambda\leq x\right\}\\
&\overset{x\in(0,4)}{=}p\sum_{k=1}^\infty \PPrb{\ceil{k\beta(x)}\leq Y_1}=p\sum_{k=1}^\infty (1-p)^{\ceil{k\beta(x)}},\\
I_p^{<,N}(x)&\coloneqq\lim_{n\to \infty}\frac{1}{L_n}\#\left\{\lambda\in\sigma\left(\bigoplus_{i=1}^n-\Delta^N_{Y_i}\right)\,\middle| \,\lambda< x\right\}\\
&\overset{x\in(0,4)}{=}p\sum_{k=1}^\infty \PPrb{\floor{(k-1)\beta(x)}< Y_1}=p\sum_{k=1}^\infty (1-p)^{\floor{(k-1)\beta(x)}+1}.
\end{align*}

Finally, our last IDS is given by
\begin{align*}
    I_p^{<,*}(x)&\coloneqq\lim_{n\to\infty}\frac{1}{L_n}\#\left\{\lambda\in\sigma\left(\bigoplus_{i=1}^n-\Delta^D_{Y_i+2}\right)\,\middle| \,\lambda< x\right\}\\
    &\overset{x\in(0,4)}{=}p\sum_{k=1}^\infty \PPrb{\floor{k\beta(x)}-2< Y_1}=p\sum_{k=1}^\infty (1-p)^{\floor{k\beta(x)}-1}.
\end{align*}

\section{Proof of Theorem 1 and Corollary 4}
\begin{proof}[Proof of Theorem 1]
We recall that $I_{p,\zeta}$ is continuous and therefore can be computed by counting eigenvalues less ($<$) or less or equal ($\leq$) than $x$
\begin{equation*}
     I_{p,\zeta}(x)=\lim_{L\to \infty}\frac{1}{L}\#\left\{\lambda\in\sigma\left(-\Delta_{L}+\zeta V_p\right)\,\middle| \,\lambda\leq x\right\}=\lim_{L\to \infty}\frac{1}{L}\#\left\{\lambda\in\sigma\left(-\Delta_{L}+\zeta V_p\right)\,\middle| \,\lambda< x\right\}.
\end{equation*}
We order the eigenvalues of any self-adjoint $n$-dimensional operator $O$ increasingly allowing for multiplicities
\begin{equation*}
    \lambda_1(O)\leq\lambda_2(O)\leq\cdots\leq\lambda_n(O).
\end{equation*}
We will derive one lower bound and two upper bounds for $I_{p,\zeta}$.

The lower bound comes from $-\Delta_{L_n}+\zeta V_p\leq-\Delta_{L_n}+\infty V_p$ where the right-hand side is to be understood as the direct sum of the free Laplacian on the connected components of $0$'s of $V_p$. To make proper sense of this, we delete from $-\Delta_{L_n}+\zeta V_p$ the $j$-th row and $j$-th column for all $j\in\{1,\ldots,L_n\}$ such that $V_p(j)=1$. The resulting principal sub-matrix is unitarily equivalent to $\bigoplus_{i=1}^n-\Delta_{Y_i}$ and because of the Cauchy Eigenvalue Interlacing Theorem we have
\begin{equation*}
    -\Delta_{L_n}+\zeta V_p\leq\bigoplus_{i=1}^n-\Delta_{Y_i},\qquad\text{first }\sum_{i=1}^n Y_i\text{ eigenvalues},
\end{equation*}
which is an inequality of eigenvalues meaning $\lambda_j\big(-\Delta_{L_n}+\zeta V_p\big)\leq\lambda_j\big(\bigoplus_{i=1}^n-\Delta_{Y_i}\big)$ for $1\leq j\leq\sum_{i=1}^n Y_i$. After counting eigenvalues less or equal ($\leq$) than $x\leq4$ of both operators, dividing by $L_n$ and taking the limit, we get
\begin{equation}\label{<=}
    I_p^\leq(x)\leq I_{p,\zeta}(x),\qquad x\in[0,4].
\end{equation}

To obtain the first upper bound for $I_{p,\zeta}$ we apply the Neumann part of Dirichlet-Neumann bracketing.
We disconnect all the $Y_i$ (we use $Y_i$ to denote the set of points as well as its cardinality) and all the points where $V_p(j)=1$ at the cost of having Neumann boundary conditions, as shown in Fig. \ref{Fig3}.
\begin{figure}[H]
    \centering
    \includegraphics{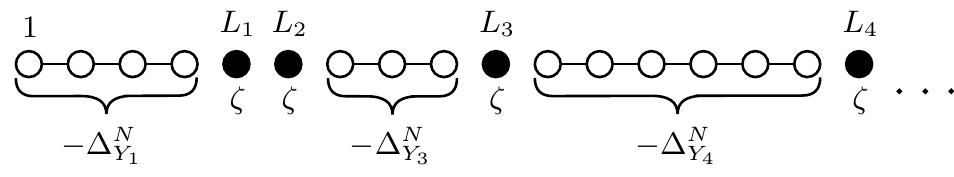}
    \caption{Resulting operator after applying the Neumann part of Dirichlet-Neumann bracketing to $-\Delta_{L_n}+\zeta V_p$ in order to disconnect all the $Y_i$'s.}
    \label{Fig3}
\end{figure}
Again, we identify unitarily equivalent operators to write
\begin{equation*}
    \left(\bigoplus_{i=1}^{n}-\Delta^N_{Y_i}\right)\oplus\zeta\id_n\leq-\Delta_{L_n}+\zeta V_p,\qquad\text{all }L_n\text{ eigenvalues}.
\end{equation*}
Counting eigenvalues less ($<$) than $x\leq4\leq\zeta$ we get
\begin{equation}\label{<N}
    I_{p,\zeta}(x)\leq I_p^{<,N}(x),\qquad x\in[0,4].
\end{equation}
Putting together \eqref{<=} and \eqref{<N} evaluated at $4-x$, subtracting $1-p$, and using \eqref{ET1} and \eqref{ET2} we conclude
\begin{align*}
    I_p^{\leq,D}(x)&\leq(1-p)-I_{p,\zeta}(4-x)\leq I_p^{<}(x),&x\in[0,4],\\
    p\sum_{k=1}^\infty (1-p)^{\ceil{k\beta(x)}}&\leq(1-p)-I_{p,\zeta}(4-x)\leq p\sum_{k=1}^\infty (1-p)^{\floor{k\beta(x)}},&x\in(0,4),
\end{align*}
which is \eqref{B2}.

To prove \eqref{B1} we need to find another upper bound for $I_{p,\zeta}$. In order to do this we construct from $-\Delta_{L_n}+\zeta V_p$ a new larger operator $-\Delta_{L_n+n}+\frac{\zeta}{2}V'$ by doubling each point in which $V_p(j)=1$ while maintaining the $Y_i$'s.
\begin{figure}[H]
    \centering
    \includegraphics{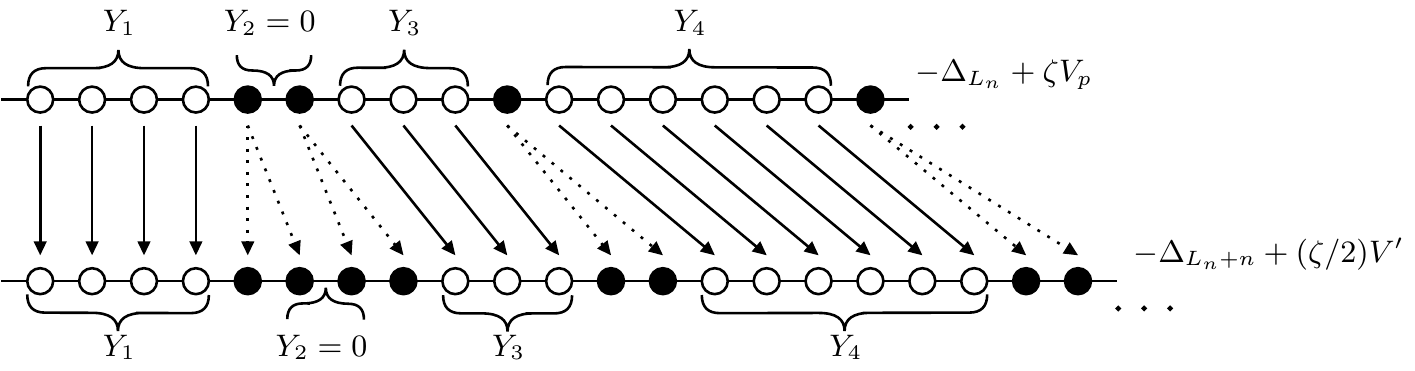}
    \caption{Construction of $-\Delta_{L_n+n}+\frac{\zeta}{2}V'$ from $-\Delta_{L_n}+\zeta V_p$ by doubling each positive potential point.}
    \label{Fig4}
\end{figure}
To be precise, $V'(j)\coloneqq\sum_{k=1}^\infty\left(\delta_{L_k+k-1,j}+\delta_{L_k+k,j}\right)$ whereas $V_p(j)=\sum_{k=1}^\infty \delta_{L_k,j}$. This doubling procedure comes with a natural linear map $T:\ell^2\left(\{1,\ldots,L_n\}\right)\longrightarrow\ell^2\left(\{1,\ldots,L_n+n\}\right)$,
\begin{equation*}
    (T\phi)(j)\coloneqq\begin{cases}\phi(j-k),&\text{if }\,L_k+k+1\leq j \leq L_{k+1}+(k+1)-1,\\
    \phi(L_k),&\text{if }\,j=L_k+k,
    \end{cases}
\end{equation*}
with the convention $L_0=0$, which assigns to $T\phi$ (lower row) the same values of $\phi$ (upper row) according to Fig \ref{Fig4}. We will now prove:
\begin{equation*}
    -\Delta_{L_n+n}+\frac{\zeta}{2}V'\leq -\Delta_{L_n}+\zeta V_p,\qquad\text{first }L_n\text{ eigenvalues}.
\end{equation*}

The map $T$ is clearly injective and by direct computation we have for all $\psi,\phi\in\ell^2\left(\{1,\ldots,L_n\}\right)$:
\begin{align*}
\hp{T\psi}{\frac{\zeta}{2}V' T\phi}&=\frac{\zeta}{2}\sum_{k=1}^n \cc{(T\psi)(L_k+k-1)}\,(T\phi)(L_k+k-1)+\frac{\zeta}{2}\sum_{k=1}^n \cc{(T\psi)(L_k+k)}\,(T\phi)(L_k+k)\\
&=\zeta\sum_{k=1}^n \cc{\psi(L_k)}\,\phi(L_k)=\hp{\psi}{\zeta V_p\phi},\\
\hp{T\psi}{\Delta_{L_n+n} T\phi}&=\sum_{k=0}^{n-1}\,\,\sum_{j=L_k+k+1}^{L_{k+1}+(k+1)-2}\cc{(T\psi)(j)}\,(\Delta_{L_n+n} T\phi)(j)\\
&\qquad+\sum_{k=1}^n \underbrace{\cc{(T\psi)(L_k+k-1)}}_{=(T\psi)(L_k+k)}\,\big[(\Delta_{L_n+n}T\phi)(L_k+k-1)+(\Delta_{L_n+n}T\phi)(L_k+k)\big]\\
&=\sum_{k=0}^{n-1}\,\,\sum_{j=L_k+1}^{L_{k+1}-1}\cc{\psi(j)}\,(\Delta_{L_n} \phi)(j)+\sum_{k=1}^{n}\cc{\psi(L_k)}\,(\Delta_{L_n}\phi)(L_k)=\hp{\psi}{\Delta_{L_n}\phi}.
\end{align*}
Let $\phi_i$ be the normalized eigenvector associated to $\lambda_i\big(-\Delta_{L_n}+\zeta V_p\big)$. Also, define the subspace $S\coloneqq\vect\{e_{L_k+k}\,|\,k=1,\ldots,n\}\subseteq\ell^2\left(\{1,\ldots,L_n+n\}\right)$ where the $e_i$ are the canonical basis, and its orthogonal complement $S^\perp\subseteq\ell^2\left(\{1,\ldots,L_n+n\}\right)$ with 
their associated orthogonal projectors $P_S,\,P_{S^\perp}$. Then, due to the injectivity of $T$ and the Min-Max Principle we have for $k\leq L_n$
\begin{align*}
    \lambda_k\left(-\Delta_{L_n+n}+\frac{\zeta}{2}V'\right)&\leq\sup_{(x_1\ldots,x_k)\in\C^k\setminus\{0\}}\frac{\hp{\sum_{i=1}^k x_i T\phi_i}{(-\Delta_{L_n+n}+\frac{\zeta}{2}V')\sum_{i=1}^k x_i T\phi_i}}{\norm{\sum_{i=1}^k x_i T\phi_i}^2}\\
    &=\sup_{(x_1\ldots,x_k)\in\C^k\setminus\{0\}}\frac{\sum_{i=1}^k \abs{x_i}^2 \lambda_i\big(-\Delta_{L_n}+\zeta V_p\big)}{\norm{P_S\sum_{i=1}^k x_i T\phi_i}^2+\norm{P_{S^\perp}\sum_{i=1}^k x_i T\phi_i}^2}\\
    &\leq\lambda_k\big(-\Delta_{L_n}+\zeta V_p\big) \sup_{(x_1\ldots,x_k)\in\C^k\setminus\{0\}}\frac{\sum_{i=1}^k \abs{x_i}^2}{\norm{\sum_{i=1}^k x_i P_{S^\perp}T\phi_i}^2}\\
    &=\lambda_k\big(-\Delta_{L_n}+\zeta V_p\big) \sup_{(x_1\ldots,x_k)\in\C^k\setminus\{0\}}\frac{\sum_{i=1}^k \abs{x_i}^2}{\norm{\sum_{i=1}^k x_i \phi_i}^2}=\lambda_k\big(-\Delta_{L_n}+\zeta V_p\big),
\end{align*}
which ends the proof of the claim.

The advantage of $-\Delta_{L_n+n}+\frac{\zeta}{2}V'$ is that it does not have isolated points with positive potential. Hence, we can apply to it the Neumann part of Dirichlet-Neumann bracketing to disconnect each $Y_i$ together with its two adjacent points, as shown in Fig. \ref{Fig5}.
\begin{figure}[H]
    \centering
    \includegraphics{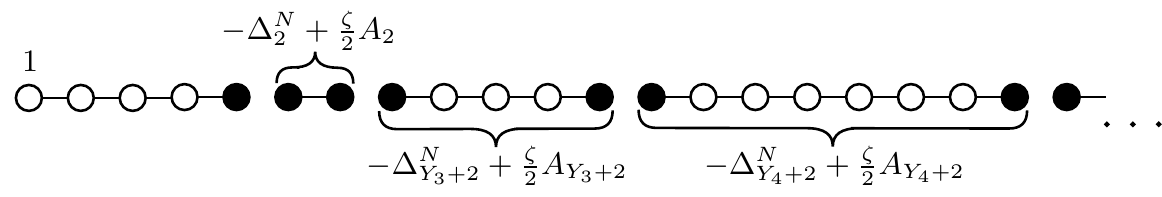}
    \caption{Resulting operator after applying the Neumann part of Dirichlet-Neumann bracketing to $-\Delta_{L_n+n}+\frac{\zeta}{2}V'$ in order to disconnect each $Y_i$ together with its two adjacent points.}
    \label{Fig5}
\end{figure}
Since $\zeta\geq4$ we have that $-\Delta^N_{n}+\frac{\zeta}{2}A_{n}\geq-\Delta^D_{n}$, and therefore
\begin{equation*}
    \left(Y_1+1\text{ term} \right)\oplus\left(\bigoplus_{i=2}^n-\Delta^D_{Y_i+2}\right)\oplus 2\leq-\Delta_{L_n+n}+\frac{\zeta}{2}V',\qquad\text{all }L_n+n\text{ eigenvalues}.
\end{equation*}
The first and last terms are different because of the boundary at $j=0$ and $j=L_n+n+1$, but they contribute at most $\bigo(\ln(n))$ eigenvalues, so they will disappear in the limit. By counting eigenvalues less ($<$) than $x\leq4$ we get
\begin{equation}\label{<*}
    I_{p,\zeta}(x)\leq I_p^{<,*}(x),\qquad x\in[0,4].
\end{equation}
Putting together \eqref{<=} and \eqref{<*} we obtain
\begin{align*}
    I_p^{\leq}(x)&\leq\,I_{p,\zeta}(x)\leq I_p^{<,*}(x)&x\in[0,4],\\
    p\sum_{k=1}^\infty (1-p)^{\ceil{k\beta(x)}-1}&\leq\,I_{p,\zeta}(x)\leq p\sum_{k=1}^\infty (1-p)^{\floor{k\beta(x)}-1}&x\in(0,4).
\end{align*}
which is \eqref{B1}. 

We have actually shown for $\zeta\geq4$
\begin{equation}\label{B3}
    I_p^{\leq}(x)\leq\,I_{p,\zeta}(x)\leq\min\{ I_p^{<,N}(x),\,I_p^{<,*}(x)\},\qquad x\in[0,4],
\end{equation}
which is equivalent to Theorem 1 since $I_p^{<,*}\leq I_p^{<,N}$ on $[0,2]$ and $I_p^{<,N}\leq I_p^{<,*}$ on $[2,4]$. We prefer the statement of Theorem 1 since it makes evident the appearance of the increasing sequence of special energies.
\end{proof}

\begin{proof}[Proof of Corollary 4]
First we remark that the IDS of the free Laplacian is $I_{0,0}(x)=I_{0,\zeta}(x)=\frac{1}{\beta(x)}$ for $x\in(0,4]$. We also recall that \eqref{B1} holds on $(0,4)$ and therefore so does \eqref{IUC}. Since $I_{p,\zeta}(x)$ is monotone decreasing in $p$ we have
\begin{equation*}
f(p,x)\coloneqq p\frac{(1-p)^{\beta(x)}}{1-(1-p)^{\beta(x)}}\leq I_{p,\zeta}(x)\leq I_{0,\zeta}(x)=\frac{1}{\beta(x)},\qquad x\in(0,4].
\end{equation*}
Notice $I_{0,0}(0)-I_{p,\zeta}(0)=0$ and $I_{0,0}(4)-I_{p,\zeta}(4)=p$, hence, we can avoid the boundary points and just show
\begin{equation*}
\sup_{x\in(0,4)}\left(\frac{1}{\beta(x)}- f(p,x)\right)\leq p.
\end{equation*}
An application of L'Hôpital's rule gives $\lim_{p\to 0}f(p,x)=\frac{1}{\beta(x)}$, which implies 
\begin{equation*}
\frac{1}{\beta(x)}- f(p,x)=-\int_0^p\frac{\partial f}{\partial t}(t,x)\dd{t}\leq p\sup_{\substack{t\in(0,1)\\y\in(0,4)}}\abs{-\frac{\partial f}{\partial t}(t,y)},\qquad x\in(0,4).
\end{equation*}
In the open set $(t,y)\in(0,1)\times(0,4)$ we have
\begin{equation*}
-\frac{\partial f}{\partial t}(t,y)=\frac{(1-t)^{\beta(y) -1} \left[(1-t)^{\beta(y) +1}-(1-t(\beta(y)+1))\right]}{\left(1-(1-t)^{\beta(y)}\right)^2}.
\end{equation*}
Using the binomial inequality we see $-\frac{\partial f}{\partial t}(t,y)\geq0$. Moreover, basic algebraic manipulation shows 
\begin{equation*}
-\frac{\partial f}{\partial t}(t,y)>1\iff 0>1-(1-t)^{\beta(y)-1}[1+t(\beta(y)-1)],
\end{equation*}
but
\begin{equation*}
    1-(1-t)^{\beta(y)-1}[1+t(\beta(y)-1)]=\beta(y)(\beta(y)-1)\int_0^t(1-s)^{\beta(y)-2}s\dd{s}\geq0.
\end{equation*}
Hence we conclude $-\frac{\partial f}{\partial t}(t,y)\leq1$.
\end{proof}

\section{Final Remarks}
In the proof of Theorem 1 we only used the hypothesis $\zeta\geq4$ for \eqref{<*}, where we needed $\frac{\zeta}{2}+1\geq3$ to obtain the Laplacians $-\Delta^D_{Y_i+2}$. In fact, for $\zeta\geq0$ \eqref{<=} holds, and instead of \eqref{<N} we could have written
\begin{equation}\label{<N2}
I_{p,\zeta}(x)\leq I_p^{<,N}(x),\qquad x\in[0,\min\{\zeta,4\}].
\end{equation}
From \eqref{<=} and \eqref{<N2} evaluated at $4-x\leq\zeta$, after subtracting $1-p$ and using \eqref{ET1} and \eqref{ET2}, we arrive at:
\begin{Corollary}
$n\in\N\setminus\{1\},\,\,\zeta\geq4-\beta^{-1}(n)\implies I_{p,\zeta}(4-\beta^{-1}(n))=(1-p)-\frac{p(1-p)^n}{1-(1-p)^n}$.
\end{Corollary}

This corollary says the function $\zeta\longmapsto I_{p,\zeta}(4-\beta^{-1}(n))$ is constant on $[4-\beta^{-1}(n),\infty)$, not just on $[4,\infty)$ as Corollary 2 stated. It is natural to conjecture the same behavior for the decreasing sequence
\begin{Conjecture}$n\in\N\setminus\{1\},\,\,\zeta\geq\beta^{-1}(n)\implies I_{p,\zeta}(\beta^{-1}(n))=\frac{p}{(1-p)}\frac{(1-p)^n}{1-(1-p)^n}$,
\end{Conjecture}
or at least the weaker
\begin{Conjecture}$\zeta\geq2\implies I_{p,\zeta}(\beta^{-1}(n))=\frac{p}{(1-p)}\frac{(1-p)^n}{1-(1-p)^n}$ for all $n\in\N\setminus\{1\}$.
\end{Conjecture}

To prove Conjecture A we would need to handle arbitrarily small $\zeta$ which seems impossible with our methods. For Conjecture B, it would be sufficient to do the case $\zeta=2$ since $I_{p,\zeta}(x)$ is monotone decreasing in $\zeta$. Following the same steps of Sec. III, we could get to
\begin{equation*}
    \left(Y_1+1\text{ term} \right)\oplus\left(\bigoplus_{i=2}^n-\Delta_{Y_i+2}\right)\oplus 1\leq -\Delta_{L_n+n}+2V_p,\qquad\text{first }L_n\text{ eigenvalues},
\end{equation*}
where the IDS of the left-most side can be explicitly computed, but is strictly greater than $I_{p}^\leq$ and therefore not enough to conclude.

As a last comment, we remark that the results obtained in this article can be extended to any 1D Anderson model with a positive potential distribution as long as $\PPrb{V=0}>0$ and $[0,4]$ is separated from the rest of the spectrum. In such a case Theorem 1 and its Corollaries hold once we replace $p\mapsto1-\PPrb{V=0}$. In particular, in $[0,4]$ we will find the same sequences of special energies at which the IDS is explicit.

\begin{acknowledgments}
This work has benefitted from support provided by the University of Strasbourg Institute for Advanced Study (USIAS), within the French national programme “Investment for the future” (IdEx-Unistra). The author thanks the referees for their careful reading and valuable suggestions.
\end{acknowledgments}

\section*{Data Availability}
Data sharing is not applicable to this article as no new data were created or analyzed in this study.

%

\end{document}